\documentclass[12pt,a4paper]{article}
\usepackage{amsmath,amsthm,amssymb,amscd}
\usepackage{latexsym}
\usepackage{indentfirst}
\usepackage{bibentry}
\usepackage{textcomp}
\usepackage{float}
\usepackage{multirow}
\usepackage{booktabs}
\usepackage{graphicx}
\usepackage{color}
\usepackage[nosort]{cite}
\usepackage{authblk}
\usepackage[numbers]{natbib}
\setcitestyle{open={},close={}}

\makeatletter \@addtoreset{equation}{section}

\makeatother
\newtheorem{thm}{Theorem}[section]

\newtheorem{lem}{Lemma}[section]
\newtheorem{exmp}{Example}[section]

\theoremstyle{definition}
\newtheorem{rem}{Remark}[section]

\begin{document}

\title{The Smallest Eigenvalue of Large Hankel Matrices}

\author[1]{{Mengkun Zhu}\footnote{Zhu\_mengkun@163.com}}
\author[1]{{Yang Chen}\footnote{yangbrookchen@yahoo.co.uk}}
\author[2]{{Niall Emmart}\footnote{nemmart@yrrid.com}}
\author[2]{{Charles Weems}\footnote{weems@cs.umass.edu}}

\affil[1]{Department of Mathematics, University of Macau,
Avenida da Universidade, Taipa, Macau, China}

\affil[2]{College of Information and Computer Sciences, \protect\\
University of Massachusetts, Amherst, MA 01003, USA}

\renewcommand\Authands{ and }

\maketitle

\begin{abstract}

We investigate the large $N$ behavior of the smallest eigenvalue, $\lambda_{N}$, of an $\left(N+1\right)\times \left(N+1\right)$ Hankel (or moments) matrix $\mathcal{H}_{N}$, generated by the weight  $w(x)=x^{\alpha}(1-x)^{\beta},~x\in[0,1],~ \alpha>-1,~\beta>-1$. By applying the arguments of Szeg\"{o}, Widom and Wilf, we establish the asymptotic formula for the orthonormal polynomials $P_{n}(z),z\in\mathbb{C}\setminus[0,1]$, associated with $w(x)$, which are required in the determination of $\lambda_{N}$.
Based on this formula, we produce the expressions for $\lambda_{N}$, for large $N$.

Using the parallel algorithm presented by Emmart, Chen and Weems, we show that the theoretical results are in close proximity to the numerical results for sufficiently large $N$.
\end{abstract}


\section{Introduction}

Let $\mu(x)$ be a positive measure with the bounded support $I(\subseteq\mathbb{R})$ and define the moment sequence of $\mu(x)$ by
\begin{equation}\label{b2}
h_{k}:=\int_{I}x^{k}d\mu(x) ~~~~~~~k=0,1,2,\cdots.
\end{equation}

We know that the Hankel determinant plays a significant role in the theory of random Hermitian matrices. Associated with $\mu(x)$, the $\left(N+1\right)\times\left(N+1\right)$ Hankel matrix $\mathcal{H}_{N}$, is defined by
\begin{equation}\label{b1}
\mathcal{H}_{N}:=\left(h_{m+n}\right)_{m,n=0}^{N}.
\end{equation}

It is known that the smallest eigenvalue of the Hankel matrix is intimately related to the distribution function $\mu(x)$. We are motivated by the fact that the smallest eigenvalue depends $\mu'(x)$ in a non-trivial way.

\par
Let $I=[a,b]$, where $a$ and $b$ are fixed constants, such that the Szeg\"{o} condition,
\begin{equation}
\int_{a}^{b}\frac{\ln w(x)}{\sqrt{(b-x)(x-a)}}dx>-\infty,
\end{equation}
with $w(x)=\mu'(x)$ is satisfied. The asymptotic behavior of the Hankel determinants for large enough $N$ is given by Szeg\"{o} [\cite{C2}-\cite{C4}].

Let $\lambda_{N}$ denote the smallest eigenvalue of $\mathcal{H}_{N}$. The behavior of $\lambda_{N}, N$ large, has attracted a lot of attention. See e.g. Szeg\"{o} [\cite{C3}], Widom and Wilf [\cite{C6}], Chen \emph{et al.} [\cite{C9,C10,C13}], Berg \emph{et al.} [\cite{C13,C14}], etc. Szeg\"{o} [\cite{C3}] studied the special cases for $w(x)$, defined on $I$, which can either be a finite or infinite. For finite cases, if $w(x)=1,x\in(-1,1)$ or $w(x)=1,x\in(0,1)$, the smallest eigenvalues for large $N$ are given, respectively, by\footnote[1]{ Throughout this paper, the relation $a_{n}\simeq b_{n}$ means $\lim_{n\rightarrow\infty}a_{n}/b_{n}$=1.}

\begin{equation*}
~\lambda_{N}~\simeq~2^{\frac{9}{4}}\pi^{\frac{3}{2}}\sqrt{N}\left(\sqrt{2}+1\right)^{-2N-3},
\end{equation*}
\begin{equation*}
\lambda_{N}~\simeq~2^{\frac{15}{4}}\pi^{\frac{3}{2}}\sqrt{N}\left(\sqrt{2}+1\right)^{-4N-4}.
\end{equation*}
Widom and Wilf [\cite{C6}] found a kind of `universal' law, where they show that if $w(x)>0,x\in[a,b]$, and
the Szeg\"{o} condition (1.3) is satisfied, then
\begin{equation*}
\lambda_{N}~\simeq~ A\sqrt{N}B^{-N},
\end{equation*}
where $A$ and $B$ are computable constants depending on $w(x)$, $a$, $b$, and are independent of $N$.
\par
For cases of an infinite interval, Szeg\"{o} [\cite{C3}] chose the Gaussion weight ($w(x)={\rm e}^{-x^{2}},~x\in\mathbb{R}$) and Laguerre weight ($w(x)={\rm e}^{-x},~x\geq0$). The corresponding smallest eigenvalues are approximated, respectively, by
\begin{align*}
&\lambda_{N}~\simeq~{\rm e}2^{\frac{13}{4}}\pi^{\frac{3}{2}}N^{\frac{1}{4}}{\rm e}^{-2\sqrt{2N}}, \\
&\lambda_{N}~\simeq~{\rm e}2^{\frac{7}{2}}\pi^{\frac{3}{2}}N^{\frac{1}{4}}{\rm e}^{-4\sqrt{N}}.
\end{align*}

Chen and Lawrence generalized the results of Szeg\"{o} in [\cite{C3}]. By means of Dyson's Coulomb fluid method, they deduced the case for $w(x)={\rm e}^{-x^{\beta}}, ~x\in[0,+\infty), ~\beta>\frac{1}{2}$ and then gave two asymptotic formulas of $\lambda_{N}$ for $\beta=n+1/2$ and $\beta\neq n+1/2$, $n=1,2,3,\cdots$, respectively. See [\cite{C9}] for details.

We note that the smallest eigenvalues of the examples given above are exponentially small. Hence, it's hard to determine the smallest eigenvalues of the Hankel matrices associated with these weights by numerical techniques.

This paper is organized as follows, firstly, we establish the asymptotic formula for the orthonormal polynomials $P_{n}(z)$ associated with the weight $w(x)$ in Theorem 2.1. Then in Theorem 2.2, we give the specific asymptotic expression of $\lambda_{N}$. Finally, we present some numerical results compared with our theoretical results in section 3.

In order to meet the demands of some proofs in our results, we define the whole complex plane by $\mathbb{C}\cup\{\infty\}$, and the unit disc by
\begin{equation*}
D:=\{z\in\mathbb{C}\big||z|\leq1\},
\end{equation*}
with its boundary (unit circle)
\begin{equation*}
\partial D:=\{z\in\mathbb{C}\big||z|=1\}.
\end{equation*}

\section{Main results}
In this section, we shall produce the asymptotic expression for $\lambda_{N}$, the smallest eigenvalue of the $\left(N+1\right)\times \left(N+1\right)$ Hankel matrix $\mathcal{H}_{N}$. We consider the weight
\begin{equation}
w(x)=x^{\alpha}(1-x)^{\beta}, x\in[0,1], ~\alpha>-1,~\beta>-1,
\end{equation}
which satisfies
\begin{equation*}
\int_{0}^{1}\frac{\ln w(x)}{\sqrt{x(1-x)}}dx=-2\pi(\alpha+\beta)\ln2>-\infty.
\end{equation*}
The $N+1$ by $N+1$ Hankel matrix $\mathcal{H}_{N}$ is defined by
\begin{equation*}
\mathcal{H}_{N}:=\left(h_{m+n}\right)_{m,n=0}^{N},
\end{equation*}
where $h_{m+n}$ is the $(m+n)$th moment with respected to $w(x)$, reads
\begin{equation*}
h_{m+n}:=\int_{0}^{1}x^{m+n}w(x)dx=\int_{0}^{1}x^{m+n+\alpha}(1-x)^{\beta}dx,~ m,n=0,1,2,\cdots.
\end{equation*}

By the definition of the Gamma function
\begin{equation*}
\Gamma(x)=\int_{0}^{\infty}t^{x-1}{\rm e}^{-t}dt, ~~\Re~x>0,
\end{equation*}
and the Beta function
\begin{equation*}
B(x,y)=\int_{0}^{1}t^{x-1}(1-t)^{y-1}dt, ~~\Re~x>0, \Re~y>0,
\end{equation*}
with the relationship
\begin{equation*}
B(x,y)=\frac{\Gamma(x)\Gamma(y)}{\Gamma(x+y)},
\end{equation*}
we have
\begin{equation}
h_{m+n}=\frac{\Gamma(\alpha+m+n+1)\Gamma(\beta+1)}{\Gamma(\alpha+\beta+m+n+2)},~~m,n=0,1,2,\cdots,N.
\end{equation}

The Hankel matrix for $\alpha=\beta=0$ is the Hilbert matrix $\mathcal{H}_{N}=\left(\frac{1}{m+n+1}\right)_{m,n=0}^{N}$, for which some partial results were obtained in [\cite{C3}-\cite{C5}], [\cite{C6}] (in which the factor $-1/2\pi$ of Lemma 2 should be changed to $-1/4\pi$). The following two examples give $\mathcal{H}_{N}$ for some special choices of $\alpha$ and $\beta$.
\begin{exmp}
For $\alpha>-1$,and $\beta=0$, {\rm i.e.} $w(x)=x^{\alpha}, x\in[0,1]$, the Hankel matrix reads
\begin{equation*}
\mathcal{H}_{N}=\left(\frac{1}{1+m+n+\alpha}\right)_{m,n=0}^{N}.
\end{equation*}
\end{exmp}
\begin{exmp}
For $\alpha=\beta=-\frac{1}{2}$, {\rm i.e.} $w(x)=\frac{1}{\sqrt{x(1-x)}}, ~x\in[0,1]$, the Hankel matrix is given by
\begin{equation*}
\mathcal{H}_{N}=\left(\frac{\sqrt{\pi}\Gamma\left(\frac{1}{2}+m+n\right)}{\Gamma\left(1+m+n\right)}\right)_{m,n=0}^{N}.
\end{equation*}
\end{exmp}

For generic $\alpha$ and $\beta$, we shall show that there is an asymptotic formula for $\lambda_{N}$, the smallest eigenvalues of the Hankel matrices, in the following form
\begin{equation*}
\lambda_{N}\simeq\frac{8\sqrt{N}}{\psi(\alpha,\beta)\left(1+\sqrt{2}\right)^{4N+2}}~,
\end{equation*}
where
\begin{equation*}
\psi(\alpha,\beta)=2^{-\frac{3}{4}}\pi^{-\frac{3}{2}}\left(1+2^{\frac{1}{2}}\right)^{2\alpha+2}\left(1+2^{-\frac{1}{2}}\right)^{2\beta}.
\end{equation*}
See details in the proof of the Theorem \ref{th2}.

Let $\left\{P_{n}(x)\right\}_{n=0}^{\infty}$ be the orthonormal polynomials associated with our weight $w(x)$, i.e.,
\begin{equation*}
\int_{0}^{1}P_{m}(x)P_{n}(x)w(x)dx=\delta_{m,n},~~m,n=0,1,\cdots,N.
\end{equation*}

We define $P_{n}(x)$ and the $k$th moment $h_{k}$ of $w(x)$ to be
\begin{equation*}
P_{n}(x):=\sum_{k=0}^{n}a_{n,k}x^{k},
\end{equation*}
\begin{equation*}
h_{k}:=\int_{0}^{1}x^{k}w(x)dx.
\end{equation*}
Then, the orthogonality relation can be rewritten as
\begin{equation*}
\delta_{m,n}=\sum_{i,j=0}^{N}a_{m,i}h_{i+j}a_{n,j},~~m,n=0,1,\cdots,N,
\end{equation*}
which, in matrix form, reads
\begin{equation}\label{1}
\mathcal{I}=\mathcal{A}_{N}\mathcal{H}_{N}\mathcal{A}_{N}^{T}~,
\end{equation}
where
\begin{gather*}
\mathcal{A}_{N}:=\begin{bmatrix}
a_{0,0} & 0 & 0 & \cdots & 0 \\
a_{1,0} & a_{1,1} & 0 & \cdots & 0 \\
a_{2,0} & a_{2,1} & a_{2,2} & \cdots & 0 \\
\cdots & \cdots & \cdots & \cdots &\cdots \\
a_{N,0} & a_{N,1} & a_{N,2} & \cdots & a_{N,N} \\
\end{bmatrix}.\quad
\end{gather*}

From (\ref{1}), we find
\begin{equation*}
\mathcal{H}_{N}^{-1}=\mathcal{A}_{N}^{T}\left(\mathcal{A}_{N}\mathcal{A}_{N}^{T}\right)\left(\mathcal{A}_{N}^{T}\right)^{-1},
\end{equation*}
which shows that $\frac{1}{\lambda_{N}}$ is the largest eigenvalue of $\mathcal{A}_{N}\mathcal{A}_{N}^{T}$. This is not a new result, for more details, see [\cite{C3,C4,C6}]. Denoting by $\sigma_{m,n}$ the $(m,n)$th entry of $\mathcal{A}_{N}\mathcal{A}_{N}^{T}$, we have
\begin{equation*}
\sigma_{m,n}:=\left(\mathcal{A}_{N}\mathcal{A}_{N}^{T}\right)_{m,n}=\sum_{k=0}^{N}a_{m,k}a_{n,k}=\frac{1}{2\pi}\int_{0}^{2\pi}P_{m}\left({\rm e}^{{\rm i}\theta}\right)P_{n}\left({\rm e}^{-{\rm i}\theta}\right)d\theta.
\end{equation*}

We shall make use of $P_{n}(z)$ to study the behavior of $\sigma_{m,n}$ and thus of $\lambda_{N}$.
 Szeg\"{o} [\cite{C2}] has proved the case of the interval $[-1,1]$ if $\mu(x)$ is absolutely continuous and, Geronimus has proved that case in [\cite{C1}] for general $\mu(x)$. We can deduce the case $[0,1]$ by a linear transformation $\varphi: x\mapsto \frac{1}{2}(x+1), x\in [-1,1]$. Since if $\Phi_{n}(x)$ are the orthonormal polynomials associated with the weight $w(x)=\mu'\left(\varphi(x)\right)$ then $P_{n}(x)=\Phi_{n}(2x-1),x\in[0,1]$. Based on the discussion in [\cite{C1}, Thm.9.3], [\cite{C2}, Thm.12.1.2] and [\cite{C6}, Lem.2], etc., we obtain an asymptotic expression for $P_{n}(z),~z\in\mathbb{C}\setminus[0,1]$.
\begin{thm}\label{th1}
The asymptotic behavior of the orthonormal polynomials $P_{n}(z)$ with respect to the weight $w(x)=x^{\alpha}(1-x)^{\beta}, ~x\in[0,1],~ \alpha>-1,~\beta>-1$, uniformly for $z$ on compact subsets of $\mathbb{C}\setminus[0,1]$, satisfies
\begin{equation*}
P_{n}(z)\simeq\frac{1}{\sqrt{\pi}}\zeta^{n}(z)A(\zeta(z)),
\end{equation*}
Here
\begin{equation}\label{2.4}
\zeta(z)=\left(\sqrt{z}+\sqrt{z-1}\right)^{2},
\end{equation}
with the square roots taking the positive values as $\Re z\rightarrow\infty$, and noting that $|\zeta(z)|>1$ in $z\in\mathbb{C}\setminus[0,1]$, then $|A(\zeta(z))|$ is given by
\begin{equation}\label{2.5}
\left|A\left(r{\rm e}^{{\rm i}\theta}\right)\right |={ \exp\left[\frac{1}{4\pi}\int_{-\pi}^{\pi}\ln\left({\cos^{2\alpha}\frac{t}{2}}{\sin^{2\beta}\frac{t}{2}}|\sin t|\right)\frac{1-r^{2}}{1-2r\cos(\theta-t)+r^{2}}dt\right]}.
\end{equation}
\end{thm}
\begin{proof}
See [\cite{C6}, Lem. 2].
\end{proof}
\begin{lem}\label{lem1}
The maximum of $\left|\zeta\left({\rm e}^{{\rm i}\theta}\right)\right|,$ for $0\leq \theta<2\pi,$
is attained at $\theta=\pi,$ and
\begin{equation*}
\max_{0\leq\theta<2\pi}\left|\zeta\left({\rm e}^{{\rm i}\theta}\right)\right|=\left(1+\sqrt{2}\right)^{2}~.
\end{equation*}
\end{lem}

\begin{proof}
From (\ref{2.4}), we have
\begin{equation*}
\left|\zeta(z)\right|=|z|+|z-1|+2\Re\left\{\sqrt{z\left(\overline{z}-1\right)}\right\},
\end{equation*}
substituting $z={\rm e}^{{\rm i}\theta}, \theta\in[0,2\pi]$, we find,
\begin{equation}\label{2.6}
\left|\zeta\left({\rm e}^{{\rm i}\theta}\right)\right|=1+\sqrt{2-2\cos\theta}+\sqrt{2-2\cos\theta+2\sqrt{2-2\cos\theta}},
\end{equation}
the upper bound for $\left|\zeta\left({\rm e}^{{\rm i}\theta}\right)\right|$ follows immediately.
\end{proof}

\begin{lem}\label{lem2}
The entries of $A_{N}A_{N}^{T}$, $\{\sigma_{m,n}\}_{m,n=0}^{N}$, have the following upper bound
\begin{equation*}
\left|\sigma_{m,n}\right|\leq C\cdot\frac{|1+\sqrt{2}|^{2(m+n)}}{\sqrt{m+n+1}}.
\end{equation*}
\end{lem}

\begin{proof}
Recall that
\begin{equation*}
\sigma_{m,n}=\frac{1}{2\pi}\int_{0}^{2\pi}P_{m}\left({\rm e}^{{\rm i}\theta}\right)P_{n}\left({\rm e}^{-{\rm i}\theta}\right)d\theta.
\end{equation*}

According to Theorem \ref{th1}, the asymptotic behavior of $\sigma_{m,n}$ depends on the factor $\zeta^{n}(z)$. Note that $\partial D\cap[0,1]=1$, so for any $\varepsilon>0$, we have
\begin{equation}\label{2.7}
\begin{split}
\left|\sigma_{m,n}\right|   &\leq \frac{1}{2\pi}\int_{-\varepsilon}^{\varepsilon}\left|P_{m}\left({\rm e}^{{\rm i}\theta}\right)P_{n}\left({\rm e}^{{\rm i}\theta}\right)
\right|d\theta+\frac{1}{2\pi}\int_{\varepsilon}^{2\pi-\varepsilon}\left|P_{m}\left({\rm e}^{{\rm i}\theta}\right)P_{n}\left({\rm e}^{{\rm i}\theta}\right)\right|d\theta\\
&\leq \frac{1}{2\pi}\int_{-\varepsilon}^{\varepsilon}\left|P_{m}\left({\rm e}^{{\rm i}\theta}\right)P_{n}\left({\rm e}^{{\rm i}\theta}\right)
\right|d\theta+ C_{0}\cdot\int_{\varepsilon}^{2\pi-\varepsilon}\left|\zeta\left({\rm e}^{{\rm i}\theta}\right)\right|^{m+n}d\theta\\
                  &= \frac{1}{2\pi}\int_{-\varepsilon}^{\varepsilon}\left|P_{m}\left({\rm e}^{{\rm i}\theta}\right)P_{n}\left({\rm e}^{{\rm i}\theta}\right)
\right|d\theta+ C_{0}\cdot\int_{\varepsilon}^{2\pi-\varepsilon}{\rm e}^{-(m+n)\left(-\ln\left|\zeta\left({\rm e}^{{\rm i}\theta}\right)\right|\right)}d\theta.
\end{split}
\end{equation}
Here, we have used $P_{n}\left({\rm e}^{{-\rm i}\theta}\right)=\overline{P_{n}\left({\rm e}^{{\rm i}\theta}\right)}$, and $C_{0}$ is a constant.

We first deal with the second integral of (\ref{2.7}). Since $\frac{d}{d\theta}\left|\zeta\left({\rm e}^{{\rm i}\theta}\right)\right|\Big|_{\theta=\pi}=0$, $\frac{d}{d\theta}\ln\left|\zeta\left({\rm e}^{{\rm i}\theta}\right)\right|\Big|_{\theta=\pi}=0$. Applying the formula (\ref{2.6}), we have
\begin{equation}\label{2.8}
\frac{d^{2}}{d\theta^{2}}\left|\zeta\left({\rm e}^{{\rm i}\theta}\right)\right|     =-\frac{1}{2}\sin\frac{\theta}{2}+\frac{1+\cos\theta\left(3+\csc\frac{\theta}{2}\right)}{4\sqrt{\sin^{2}\frac{\theta}{2}+\sin\frac{\theta}{2}}}
                                                     -\frac{\left(2+\sqrt{2}\right)\sin\theta\left(\sin\theta+\cos\frac{\theta}{2}\right)}{8\left(\sin^{2}\frac{\theta}{2}+\sin\frac{\theta}
                                                     {2}\right)^{\frac{3}{2}}},
\end{equation}
and it is immediate that,
\begin{equation*}
-\frac{d^{2}}{d\theta^{2}}\ln\left|\zeta\left({\rm e}^{{\rm i}\theta}\right)\right|\Big|_{\theta=\pi}=\frac{\sqrt{2}}{8}>0.
\end{equation*}
Applying the Laplace method when $m+n$ large enough, we get

\begin{equation*}
C_{0}\cdot\int_{\varepsilon}^{2\pi-\varepsilon}{\rm e}^{-(m+n)\left(-\ln\left|\zeta\left({\rm e}^{{\rm i}\theta}\right)\right|\right)}d\theta\leq C_{1}\cdot\frac{\left|\zeta\left(-1\right)\right|^{m+n}}{\sqrt{m+n}},
\end{equation*}
where $C_{1}$ is also a constant. So for all $m,n$ there is another suitable constant $C$ such that
\begin{equation}\label{2.9}
C_{1}\cdot\frac{\left|\zeta\left(-1\right)\right|^{m+n}}{\sqrt{m+n}}\leq C\cdot\frac{\left|\zeta\left(-1\right)\right|^{m+n}}{\sqrt{m+n+1}}, ~~m,n=0,1,2\cdots.
\end{equation}

 To estimate the first integral in (\ref{2.7}), let $R_{\varepsilon}$ be a rectangle with its four vertices $\cos\varepsilon\pm {{\rm i}}\sin\varepsilon$, $1\pm {{\rm i}}\tan\varepsilon$. The arc of $\partial D$ given by $|\theta|\leq\varepsilon$ is contained within $R_{\varepsilon}$. Applying the Theorem 3.3.1 of [\cite{C2}], the polynomials $P_{m}(z)P_{n}(z)$ has only real zeros, so its maximum absolute value on $R_{\varepsilon}$ must be attained on the horizontal sides of $R_{\varepsilon}$. Hence, according to the Theorem \ref{th1}, we have
\begin{equation*}
\limsup_{m+n\rightarrow+\infty}\max_{z\in R_{\varepsilon}}\left|P_{m}(z)P_{n}(z)\right|^{1/(m+n)}=\left|\zeta\left({\rm e}^{{{\rm i}}\left(\varepsilon+O\left(\varepsilon^{2}\right)\right)}\right)\right|.
\end{equation*}

So as $m+n\rightarrow\infty$,
\begin{equation*}
\int_{-\varepsilon}^{\varepsilon}\left|P_{m}\left({\rm e}^{{\rm i}\theta}\right)P_{n}\left({\rm e}^{{\rm i}\theta}\right)\right|d\theta=O\left(\left|\zeta\left({\rm e}^{{{\rm i}}\left(2\varepsilon\right)}\right)\right|^{m+n}\right),
\end{equation*}
since $\left|\zeta\left({\rm e}^{{{\rm i}}\left(2\varepsilon\right)}\right)\right|>\left|\zeta\left({\rm e}^{{{\rm i}}\left(\varepsilon+O\left(\varepsilon^{2}\right)\right)}\right)\right|$ if $\varepsilon\rightarrow0$. Hence, as $\varepsilon\rightarrow0$, we have
\begin{equation}\label{2.10}
\int_{-\varepsilon}^{\varepsilon}\left|P_{m}\left({\rm e}^{{\rm i}\theta}\right)P_{n}\left({\rm e}^{{\rm i}\theta}\right)\right|d\theta=o\left(\frac{\left|\zeta\left(-1\right)\right|^{m+n}}{\sqrt{m+n}}\right),
~m+n\rightarrow+\infty,
\end{equation}
since by Lemma 2.1, $\zeta\left({\rm e}^{{\rm i}\theta}\right)$ attains its maximum modulus at $\theta=\pi,$ and not at $\theta=0$. Thus the Lemma \ref{lem2} is proved based on formulas (\ref{2.7}), (\ref{2.9}) and (\ref{2.10}).
\end{proof}

\begin{rem}\label{20}
The Laplace method [\cite{C27}, p.96, Problem 201] gives,
\begin{equation*}
\int_{a}^{b}f(t){\rm e}^{-\lambda g(t)}dt\simeq {\rm e}^{-\lambda g(c)}f(c)\sqrt{\frac{2\pi}{\lambda g''(c)}}~,~~~~{\rm as}~~\lambda\rightarrow\infty,
\end{equation*}
where $g$ assumes a strict minimum over $[a,b]$ at an interior critical point $c$, such that
\begin{equation*}
\begin{split}
g'(c)=0,~~~~~g''(c)>0~~~~~{\rm and}~~~~~f(c)\neq0.\\
\end{split}
\end{equation*}
\end{rem}

\begin{lem}
For $z={\rm e}^{{\rm i}\pi}=-1$, we find
\begin{equation}
\begin{split}\label{2.24}
\left|A(\zeta(-1))\right|&=2^{-\frac{3}{4}}\left(1+2^{\frac{1}{2}}\right)^{\alpha+1}\left(1+2^{-\frac{1}{2}}\right)^{\beta}.
\end{split}
\end{equation}
\end{lem}
\begin{proof}
From (\ref{2.5}), we obtain
\begin{equation*}
\begin{split}
\ln\left|A(\zeta(-1))\right| &=\frac{1}{4\pi}\int_{-\pi}^{\pi}\ln\left({\cos^{2\alpha}\frac{t}{2}}{\sin^{2\beta}\frac{t}{2}}|\sin t|\right)\frac{1-\eta^{2}}{1-2\eta\cos(\pi-t)+\eta^{2}}dt\\
&=\frac{\ln2}{2\pi}\int_{0}^{\pi}\frac{1-\eta^{2}}{1+2\eta\cos t+\eta^{2}}dt~~~~~~~~~~~~~~~~~~~~~~~~~~~~~~~~~~~~~~~~~~~~~~(~a~)\\
&+\frac{2\alpha+1}{2\pi}\int_{0}^{\pi}\ln\left(\cos\frac{t}{2}\right)\frac{1-\eta^{2}}{1+2\eta\cos t+\eta^{2}}dt~~~~~~~~~~~~~~~~~~~~~~~~~~~~(~b~)\\
&+\frac{2\beta+1}{2\pi}\int_{0}^{\pi}\ln\left(\sin\frac{t}{2}\right)\frac{1-\eta^{2}}{1+2\eta\cos t+\eta^{2}}dt,~~~~~~~~~~~~~~~~~~~~~~~~~~~~(~c~)
\end{split}
\end{equation*}
where $\eta:=\left(1+\sqrt{2}\right)^2.$
Note that if $\alpha=\beta=-\frac{1}{2}$, then (b) and (c) are zero, this is Example 2.2.
Applying the Residue theorem, the integral $(a)$ can be rewritten as
\begin{equation}\label{2.18}
\frac{\ln2}{4\pi {\rm i}}\int_{|z|=1}\frac{1-\left(1+\sqrt{2}\right)^{4}}{\left(1+\sqrt{2}\right)^{2}\left[z+\left(1+\sqrt{2}\right)^{-2}\right]
\left[z+\left(1+\sqrt{2}\right)^{2}\right]}dz=-\frac{\ln2}{2}.
\end{equation}
For the integral $(b)$, we have
\begin{equation}\label{2.26}
\begin{split}
&\frac{1}{2\pi}\int_{0}^{\pi}\ln\left(\cos\frac{t}{2}\right)\frac{1-\eta^{2}}{1+2\eta\cos t+\eta^{2}}dt\\
&=\frac{1-\eta^{2}}{4\pi}\int_{0}^{\pi}\frac{\ln\left(\frac{1+\cos t}{2}\right)}{1+2\eta\cos t+\eta^{2}}dt~~~~~~~~~~~~~~~~~~~~~~~~~~~ (x:=\cos t)\\
&=\frac{1-\eta^{2}}{4\pi}\int_{-1}^{1}\frac{\ln\left(\frac{1+x}{2}\right)}{1+2\eta x+\eta^{2}}\frac{dx}{\sqrt{1-x^{2}}}~~~~~~~~~~~~~~~~~~~~~(y:=\frac{1+x}{2})\\
&=\frac{1-\eta^{2}}{4\pi}\int_{0}^{1}\frac{\ln y}{1+2\eta(2y-1)+\eta^{2}}\frac{dy}{\sqrt{y(1-y)}}~~~~~(\eta=(1+\sqrt{2})^{2})\\
&=\frac{\sqrt{2}}{4\pi}\int_{0}^{1}\frac{\ln (1-y)}{\sqrt{y(1-y)}(y-2)}dy~~~~~~~~~~~~({\rm See~the~ Appendix})\\
&=\frac{\ln\left(1+\sqrt{2}\right)}{2}~.
\end{split}
\end{equation}
For the integral $(c)$, we find
\begin{equation}\label{2.27}
\begin{split}
&\frac{1}{2\pi}\int_{0}^{\pi}\ln\left(\sin\frac{t}{2}\right)\frac{1-\eta^{2}}{1+2\eta\cos t+\eta^{2}}dt\\
&=-\frac{\sqrt{2}}{4\pi}\int_{0}^{1}\frac{\ln (1-y)}{\sqrt{y(1-y)}(y+1)}dy\\
&=\frac{\ln\left(1+\sqrt{2}\right)}{2}-\frac{\ln2}{4}~.
\end{split}
\end{equation}

So from (\ref{2.18}), (\ref{2.26}) and (\ref{2.27}), we have
\begin{equation*}
\begin{split}
\left|A(\zeta(-1))\right|&=\exp\left[-\frac{\ln2}{2}+\frac{(2\alpha+1)\ln\left(1+\sqrt{2}\right)}{2}+\frac{(2\beta+1)\ln\frac{\left(1+\sqrt{2}\right)^{2}}{2}}{4}\right]\\
&=2^{-\frac{3}{4}}\left(1+2^{\frac{1}{2}}\right)^{\alpha+1}\left(1+2^{-\frac{1}{2}}\right)^{\beta}.
\end{split}
\end{equation*}
\end{proof}

\begin{thm}\label{th2}
For $\lambda_{N}$, we have
\begin{equation*}
\lambda_{N}\simeq 2^{\frac{15}{4}}\pi^{\frac{3}{2}}\left(1+2^{\frac{1}{2}}\right)^{-2\alpha}\left(1+2^{-\frac{1}{2}}\right)^{-2\beta}
N^{\frac{1}{2}}\left(1+2^{\frac{1}{2}}\right)^{-4\left(N+1\right)}
 .
\end{equation*}
\end{thm}
\begin{proof}
Based on the discussion in Theorem \ref{th1} and Lemma \ref{lem2}, we find
\begin{equation*}
\begin{split}
\sigma_{m,n}   =& \frac{1}{2\pi}\int_{-\varepsilon}^{\varepsilon}P_{m}\left({\rm e}^{{\rm i}\theta}\right)P_{n}\left({\rm e}^{-{\rm i}\theta}\right)d\theta+\frac{1}{2\pi}\int_{\varepsilon}^{2\pi-\varepsilon}P_{m}\left({\rm e}^{{\rm i}\theta}\right)P_{n}\left({\rm e}^{-{\rm i}\theta}\right)d\theta\\
         =&\frac{1}{2\pi^{2}}\int_{\varepsilon}^{2\pi-\varepsilon}\left|\zeta\left({\rm e}^{{\rm i}\theta}\right)\right|^{m+n}\left[{\rm sgn}~\zeta\left({\rm e}^{{\rm i}\theta}\right)\right]^{m-n}\left|A\left(\zeta\left({\rm e}^{{\rm i}\theta}\right)\right)\right|^{2}
         d\theta\\
         &+o\left(\frac{\left|\zeta(-1)\right|^{m+n}}{\sqrt{m+n}}\right), ~(m,n\rightarrow+\infty)
\end{split}
\end{equation*}
where ${\rm sgn}(z):=\frac{z}{|z|},z\in\mathbb{C}$. It should now be easy to determine the asymptotic behavior of the entries, $\{\sigma_{m,n}\}_{m,n=0}^{N}$, as $m,n\rightarrow+\infty$ with $m-n$ bounded. We know that the maximum of $\left|\zeta\left({\rm e}^{{\rm i}\theta}\right)\right|$ occurs at $\theta=\pi$, and by the Laplace method for asymptotic expansion of an integral, combined with Lemma 2.1., we get
\begin{equation}\label{2.11}
\begin{split}
\sigma_{m,n}&\simeq\frac{(-1)^{m-n}\left|A\left(\zeta\left(-1\right)\right)\right|^{2}\left|\zeta\left(-1\right)\right|
^{m+n}}{\sqrt[4]{2}\sqrt{\pi^{3}}\sqrt{m+n}}\\
&=(-1)^{m-n}2^{-\frac{1}{4}}\pi^{-\frac{3}{2}}\left|A\left(\zeta\left(-1\right)\right)\right|^{2}\left(1+\sqrt{2}\right)^{2(m+n)}(m+n)^{-\frac{1}{2}}.
\end{split}
\end{equation}
where $m,n\rightarrow+\infty$ with $m-n$ bounded.

We will now find the behavior of the eigenvalue $\lambda_{N}$, for large $N$.
\\
Let
\begin{equation*}
g(\theta):=\left|\zeta\left({\rm e}^{{\rm i}\theta}\right)\right|~~{\rm and}~~ \psi(\alpha,\beta):=\frac{\left|A\left(\zeta\left(-1\right)\right)\right|^{2}\sqrt{\eta}}{\sqrt{2}\sqrt{\pi^{3}}\sqrt{\left|
g''(\pi)\right|}}.
\end{equation*}
From (\ref{2.4}), we can get
\begin{equation*}
\zeta(-1)=-\left(1+\sqrt{2}\right)^{2}~~~{\rm and}~~~\frac{d^{2}g(\theta)}{d\theta^{2}}\Big|_{\theta=\pi}=-\frac{1}{2}-\frac{3\sqrt{2}}{8} .
\end{equation*}
Hence, from (\ref{2.24}), and an easy computation gives
\begin{equation}\label{2.28}
\begin{split}
\psi(\alpha,\beta)&=2^{-\frac{3}{4}}\pi^{-\frac{3}{2}}\left(1+2^{\frac{1}{2}}\right)^{2\alpha+2}\left(1+2^{-\frac{1}{2}}\right)^{2\beta}.
\end{split}
\end{equation}

For the sake of the completeness, we give the standard method, following closely in the
footsteps Widom and Wilf's [\cite{C6}, P. 342, Thm.], applied to our weight function $w(x)$.
We define,
\begin{equation}\label{2.12}
\rho_{m,n}:=(-1)^{m-n}\eta^{m+n},~~\eta=\left(1+\sqrt{2}\right)^{2},
\end{equation}
and
\begin{equation}\label{2.13}
\tau_{m,n}:=\sigma_{m,n}-\left(\frac{\psi(\alpha,\beta)}{\sqrt{2N}}\right)\rho_{m,n}.
\end{equation}

Fixing an $\varepsilon$ and a sufficiently large $N_{\varepsilon}$. It follows from (\ref{2.11})
that if $m$ and $n$ are sufficiently large, but $|m-n|\leq N_{\varepsilon}$, we shall have
\begin{equation*}
\left|\sigma_{m,n}-\psi(\alpha,\beta)\frac{(-1)^{m-n}\eta^{m+n}}{\sqrt{m+n}}\right|\leq
\varepsilon\frac{\eta^{m+n}}{\sqrt{m+n}}.
\end{equation*}

Therefore if $N$ is sufficiently large and much larger than $N_{\varepsilon}$ with $N-N_{\varepsilon}\leq m,n\leq N$,
\begin{equation}\label{2.14}
\begin{split}
\left|\tau_{m,n}\right|          &  =\left|\sigma_{m,n}-\psi(\alpha,\beta) \frac{(-1)^{m-n}\eta^{m+n}}{\sqrt{2N}}\right|\\
                  & \leq\left|\sigma_{m,n}-\psi(\alpha,\beta)\frac{(-1)^{m-n}\eta^{m+n}}{\sqrt{m+n}}\right|+\psi(\alpha,\beta)
                  \eta^{m+n}\left[\frac{1}{\sqrt{m+n}}-\frac{1}{\sqrt{2N}}\right]\\
                  & \leq\varepsilon_{1}\frac{\eta^{m+n}}{\sqrt{2N-2N_{\varepsilon}}}+\varepsilon_{2}\frac{\eta^{m+n}}{\sqrt{2N-2N_{\varepsilon}}}\\  &\leq\varepsilon\frac{\eta^{m+n}}{\sqrt{N}}=\frac{\left(1+\sqrt{2}\right)^{2(m+n)}\varepsilon}{\sqrt{N}}.
\end{split}
\end{equation}
where $\varepsilon_{1}$, $\varepsilon_{2}$ are arbitrary small.

It follows from Lemma \ref{lem2} that for all $m$, $n$,
\begin{equation}\label{2.15}
\left|\tau_{m,n}\right|\leq C_{1}\frac{\left(1+\sqrt{2}\right)^{2(m+n)}}{\sqrt{m+n+1}},
\end{equation}
where $C_{1}$ is a constant. Denote by $p_{N}$ the maximum modulus of the eigenvalues of $\left(\tau_{m,n}\right)_{m,n=0}^{N}$. Then from (\ref{2.14}) and (\ref{2.15}) we obtain
\begin{equation*}
\begin{split}
p_{N} ^{2}          &  \leq\sum_{m,n=0}^{N}\tau_{m,n}^{2}\\
&\leq\frac{\varepsilon^{2}}{N}\sum_{m,n=N-N_{\varepsilon}}^{N}\eta^{2(m+n)}+\left(\sum_{m=N-N_{\varepsilon}}^{N}\sum_{n=0}^{N-N_{\varepsilon}}+
\sum_{m=0}^{N-N_{\varepsilon}}\sum_{n=0}^{N}\right)
\frac{C_{1}^{2}\eta^{2(m+n)}}{m+n+1}\\
&\leq\frac{\varepsilon^{2}\eta^{4N+4}}{\left(\eta^{2}-1\right)^{2}N}
+C_{2}\frac{\eta^{2(2N-N_{\varepsilon})}}{2N-N_{\varepsilon}},
\end{split}
\end{equation*}
where $C_{2}$ is another constant. Assuming $N_{\varepsilon}$ to be sufficiently large in comparison to $\varepsilon$, this will simply for large enough $N$ to $N_{\varepsilon}$
\begin{equation}\label{2.16}
\left|p_{N}\right|\leq\left(\frac{\varepsilon^{2}\eta^{4N+4}}{\left(\eta^{2}-1\right)^{2}N}+\frac{3\varepsilon^{2}\eta^{4N+4}}{\left(\eta^{2}-1\right)^{2}N}\right)^{\frac{1}{2}}
\leq\frac{2\varepsilon\eta^{2N+2}}{\left(\eta^{2}-1\right)\sqrt{N}}=\frac{\left(1+\sqrt{2}\right)^{4N+2}\varepsilon}
{2\sqrt{2}\sqrt{N}}.
\end{equation}
Let $\lambda_{N}^{-1}$ denote the largest eigenvalue of the matrix $\left(\sigma_{m,n}\right)_{m,n=0}^{N}$. It follows from (\ref{2.12}) , (\ref{2.13}) and (\ref{2.16}) that if $q_{N}$ is the largest eigenvalue of the matrix $\left(\rho_{m,n}\right)_{m,n=0}^{N}$, then
\begin{equation*}
\frac{q_{N}\psi(\alpha,\beta)}{\sqrt{2N}}-p_{N}\leq\frac{1}{\lambda_{N}}\leq\frac{q_{N}\psi(\alpha,\beta)}{
\sqrt{2N}}+p_{N}~,
\end{equation*}
and so for sufficiently large $N$, we have
\begin{equation*}
\frac{q_{N}\psi(\alpha,\beta)}{\sqrt{2N}}-\frac{2\varepsilon\eta^{2N+2}}{\left(\eta^{2}-1
\right)\sqrt{N}}\leq\frac{1}{\lambda_{N}}\leq\frac{q_{N}\psi(\alpha,\beta)}{\sqrt{2N}}+\frac{2\varepsilon
\eta^{2N+2}}{\left(\eta^{2}-1\right)\sqrt{N}}~,
\end{equation*}
where $\eta=\left(1+\sqrt{2}\right)^{2}$ and $\psi(\alpha,\beta)=2^{-\frac{3}{4}}\pi^{-\frac{3}{2}}\left(1+2^{\frac{1}{2}}\right)^{2\alpha+2}\left(1+2^{-\frac{1}{2}}\right)^{2\beta}$.

But $\left(\rho_{m,n}\right)_{m,n=0}^{N}$ is of rank $1$, so its only nonzero eigenvalue (and certainly the largest) is given by the trace of $\left(\rho_{m,n}\right)_{m,n=0}^{N}$, i.e.
\begin{equation*}
q_{N}=\sum_{m=0}^{N}\eta^{2m}=\frac{\eta^{2N+2}-1}{\eta^{2}-1}\simeq \frac{\left(1+\sqrt{2}\right)^{4N+2}}{4\sqrt{2}}.
\end{equation*}
Hence,
\begin{equation*}
\begin{split}
\lambda_{N}&\simeq\frac{\sqrt{2N}}{q_{N}\psi(\alpha,\beta)}\simeq\frac{\left(\eta^{2}-1\right)\sqrt{2N}}{\psi(\alpha,\beta)\eta^{2N+2}}\\
&=2^{\frac{15}{4}}\pi^{\frac{3}{2}}\left(1+2^{\frac{1}{2}}\right)^{-2\alpha}\left(1+2^{-\frac{1}{2}}\right)^{-2\beta}N^{\frac{1}{2}}\left(1+2^{\frac{1}{2}}\right)^{-4\left(N+1\right)}.
\end{split}
\end{equation*}
\end{proof}
\begin{rem}
Putting $\alpha=\beta=0$, Szeg\"{o}'s classical result [\cite{C3}] for the weight function $w(x)=1$ is recovered:
\begin{equation*}
\lambda_{N}\simeq 2^{\frac{15}{4}}\pi^{\frac{3}{2}}N^{\frac{1}{2}}\left(\sqrt{2}-1\right)^{4N+4}.
\end{equation*}
\end{rem}
\section{Comparing with numerical results}
It is well known that Hankel matrices (moment matrices) of this form are extremely ill-conditioned.  This can also been seen from the
analytic formula, where the dominant term of $\lambda_N$ is $\big(1 + \sqrt{2}\big)^{-4(N+1)}$.  Due to the ill-conditioned nature of
these matrices, standard eigensolver packages based on double precision floating values can only solve for small values of $N$, for example $N<20$, before the available precision is exhausted (53 bit in the mantissa, 11 bits in the exponent).

In [\cite{C11}], Emmart, Chen and Weems developed an efficient parallel algorithm based on arbitrary precision arithmetic and the Secant
method that can handle the extreme ill-conditioning and we employ their algorithms here for our numeric computations.  We use the
numeric results to test the convergence of our asymptotic formulas to the actual smallest eigenvalues for various $N$ and several values
of the parameters $\alpha$ and $\beta$.   Even with efficient software, the computation times for the largest size, $N=1000$, require
almost 10 hours of CPU time on a modern Core i7 processor.

Tables 1-3 give samples of numerical $\lambda_{N}$ compared with theoretical $\lambda_{N}$. The errors of $\lambda_{N}$ for some special choices of $\alpha$ and $\beta$ are illustrated by Figures 1-4.

\begin{table}[H]
\centering
\caption{Numerical vs. theoretical values of $\lambda_{N}$ for $\alpha=\beta=-\frac{1}{2}$, see Example 2.2.}

\begin{tabular}{||c|c|c|c||}
\hline
Size $N+1$& Numerical $\lambda_{N}$  & Theoretical $\lambda_{N}$ & error\\
\hline
25 & $8.0295\times10^{-36}$ &$ 7.9829\times10^{-36}$& $0.5804\%$\\
 50 & $6.0370\times10^{-74}$ &$ 6.0198\times10^{-74}$& $0.2849\%$\\
 100&$2.3866\times10^{-150}$ & $ 2.3832\times10^{-150}$& $ 0.1409\%$\\
 150&$8.1510\times10^{-227}$ & $ 8.1434\times10^{-227}$& $ 0.0932\%$\\
 200 & $2.6230\times10^{-303}$& $ 2.6212\times10^{-303}$& $ 0.0701\%$\\
 250 & $8.1711\times10^{-380}$& $ 8.1665\times10^{-380}$& $ 0.0563\%$\\
 300 & $2.4937\times10^{-456}$& $ 2.4925\times10^{-456}$& $ 0.0467\%$\\
 350 & $7.5033\times10^{-533}$& $ 7.5003\times10^{-533}$& $ 0.0400\%$\\
 400 & $2.2344\times10^{-609}$& $ 2.2336\times10^{-609}$& $ 0.0350\%$\\
 450 & $6.6016\times10^{-686}$& $ 6.5995\times10^{-686}$& $ 0.0318\%$\\
 500 & $ 1.9383\times10^{-762}$ & $ 1.9378\times10^{-762}$& $ 0.0280\%$\\
 550 & $ 5.6626\times10^{-839}$ & $ 5.6611\times10^{-839}$& $ 0.0265\%$\\
 600 & $ 1.6474\times10^{-915}$ & $ 1.6470\times10^{-915}$& $ 0.0233\%$\\
700 & $ 1.3805\times10^{-1068}$ & $ 1.3802\times10^{-1068}$& $ 0.0200\%$\\
800 & $ 1.1449\times10^{-1221}$ & $ 1.1447\times10^{-1221}$& $ 0.0175\%$\\
900 & $ 9.4211\times10^{-1375}$ & $ 9.4197\times10^{-1375}$& $ 0.0155\%$\\
1000 & $ 7.7042\times10^{-1528}$ & $ 7.7031\times10^{-1528}$& $ 0.0140\%$\\
\hline
\end{tabular}
\end{table}

\begin{figure}[H]
\centering
\includegraphics[width=0.87\textwidth]{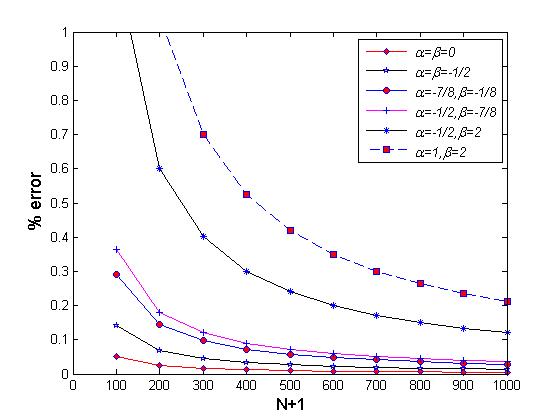}
\caption{The percentage error of the theoretical values of $\lambda_{N}$ when compared with those obtained
numerically, for various $\alpha$ and $\beta$. }
\end{figure}
\begin{figure}[H]
\centering
\includegraphics[width=0.87\textwidth]{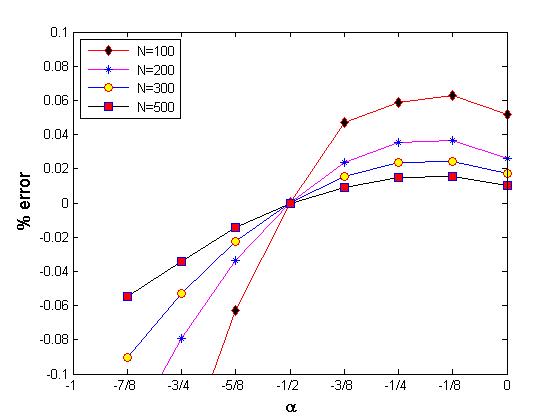}
\caption{The percentage error of the theoretical values of $\lambda_{N}$ when compared with those obtained
numerically, for different values of $N$ with $-1<\alpha<0,\beta=0$.}
\end{figure}

\begin{figure}[H]
\centering
\includegraphics[width=0.87\textwidth]{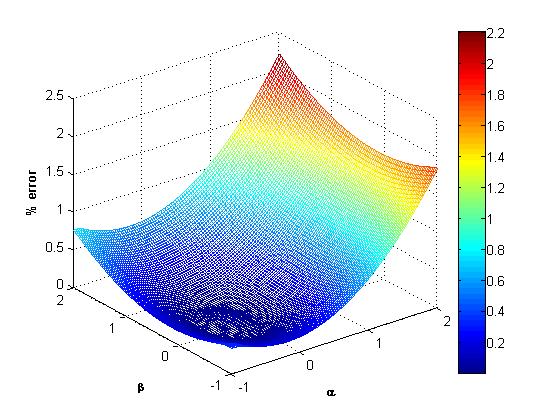}
\caption{The percentage error of the theoretical values of $\lambda_{N}, N+1=200$ when compared with those obtained
numerically, for $-1<\alpha\leq2,-1<\beta\leq2$.}
\end{figure}
\begin{figure}[H]
\centering
\includegraphics[width=0.85\textwidth]{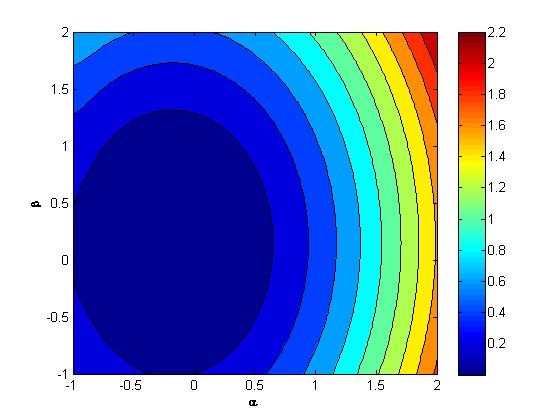}
\caption{The contour graph of the Figure $3$.}
\end{figure}

In the following, we will present some of our other numerical results.

\begin{table}[H]
\centering
\caption{List of numerical results compared with theoretical values}
\begin{tabular}{|cccllr|cccllr|}
\hline
$\alpha$        &$\beta$              &Size $N+1$      & Numerical $\lambda_{N}$ & Theoretical $\lambda_{N}$ & error\\  \hline
$ -\frac{7}{8} $&$   -\frac{1}{8}    $&  100         &$3.0997\times10^{-150}     $&$3.0907\times10^{-150}      $& $0.2911\%$\\
                &                     &  200         &$3.4042\times10^{-303}     $&$3.3993\times10^{-303}      $& $0.1449\%$\\
                &                     &  300         &$3.2355\times10^{-456}     $&$3.2324\times10^{-456}      $& $0.0965\%$\\
                &                     &  400         &$2.8988\times10^{-609}     $&$2.8967\times10^{-609}      $& $0.0723\%$\\
                &                     &  500         &$2.5144\times10^{-762}     $&$2.5130\times10^{-762}      $& $0.0578\%$\\
                &                     &  600         &$2.1369\times10^{-915}     $&$2.1359\times10^{-915}      $& $0.0482\%$\\
                &                     &  700         &$1.7906\times10^{-1068}    $&$1.7899\times10^{-1068}     $& $0.0413\%$\\
                &                     &  800         &$1.4851\times10^{-1221}     $&$1.4845\times10^{-1221}      $& $0.0361\%$\\
                &                     &  900         &$1.2220\times10^{-1374}     $&$1.2216\times10^{-1374}      $& $0.0321\%$\\
                &                     & 1000         &$9.9926\times10^{-1528}     $&$9.9897\times10^{-1528}      $& $0.0289\%$\\
$ -\frac{1}{2} $&$   -\frac{7}{8}    $&  100         &$3.5723\times10^{-150}     $&$3.5593\times10^{-150}      $& $0.3640\%$\\
                &                     &  200         &$3.9218\times10^{-303}     $&$3.9147\times10^{-303}      $& $0.1812\%$\\
                &                     &  300         &$3.7270\times10^{-456}     $&$3.7225\times10^{-456}      $& $0.1206\%$\\
                &                     &  400         &$3.3389\times10^{-609}     $&$3.3359\times10^{-609}      $& $0.0904\%$\\
                &                     &  500         &$2.8961\times10^{-762}     $&$2.8940\times10^{-762}      $& $0.0723\%$\\
                &                     &  600         &$2.4612\times10^{-915}     $&$2.4597\times10^{-915}      $& $0.0602\%$\\
                &                     &  700         &$2.0624\times10^{-1068}     $&$2.0613\times10^{-1068}      $& $0.0516\%$\\
                &                     &  800         &$1.7104\times10^{-1221}     $&$1.7096\times10^{-1221}      $& $0.0451\%$\\
                &                     &  900         &$1.4074\times10^{-1374}     $&$1.4068\times10^{-1374}      $& $0.0401\%$\\
                &                     & 1000         &$1.1508\times10^{-1527}     $&$1.1504\times10^{-1527}      $& $0.0361\%$\\
$ -\frac{1}{2} $&$         2         $&  100         &$1.6637\times10^{-151}     $&$1.6439\times10^{-151}      $& $1.1930\%$\\
                &                     &  200         &$1.8189\times10^{-304}     $&$1.8080\times10^{-304}      $& $0.5998\%$\\
                &                     &  300         &$1.7261\times10^{-457}     $&$1.7192\times10^{-457}      $& $0.4006\%$\\
                &                     &  400         &$1.5453\times10^{-610}     $&$1.5407\times10^{-610}      $& $0.3008\%$\\
                &                     &  500         &$1.3398\times10^{-763}     $&$1.3366\times10^{-763}      $& $0.2407\%$\\
                &                     &  600         &$1.1383\times10^{-916}     $&$1.1360\times10^{-916}      $& $0.2007\%$\\
                &                     &  700         &$9.5365\times10^{-1070}     $&$9.5201\times10^{-1070}      $& $0.1721\%$\\
                &                     &  800         &$7.9078\times10^{-1223}     $&$7.8959\times10^{-1223}      $& $0.1506\%$\\
                &                     &  900         &$6.5060\times10^{-1376}     $&$6.4973\times10^{-1376}      $& $0.1339\%$\\
                &                     & 1000         &$5.3197\times10^{-1529}     $&$5.3133\times10^{-1529}      $& $0.1205\%$\\
$       2      $&$   -\frac{1}{2}    $&  100         &$3.0068\times10^{-152}     $&$2.9060\times10^{-152}      $& $3.3526\%$\\
                &                     &  200         &$3.2511\times10^{-305}     $&$3.1961\times10^{-305}      $& $1.6921\%$\\
                &                     &  300         &$3.0740\times10^{-458}     $&$3.0392\times10^{-458}      $& $1.1316\%$\\
                &                     &  400         &$2.7469\times10^{-611}     $&$2.7235\times10^{-611}      $& $0.8500\%$\\
                &                     &  500         &$2.3790\times10^{-764}     $&$2.3628\times10^{-764}      $& $0.6807\%$\\
                &                     &  600         &$2.0197\times10^{-917}     $&$2.0082\times10^{-917}      $& $0.5676\%$\\
                &                     &  700         &$1.6911\times10^{-1070}     $&$1.6829\times10^{-1070}      $& $0.4867\%$\\
                &                     &  800         &$1.4018\times10^{-1223}     $&$1.3958\times10^{-1223}      $& $0.4260\%$\\
                &                     &  900         &$1.1529\times10^{-1376}     $&$1.1486\times10^{-1376}      $& $0.3788\%$\\
                &                     & 1000         &$9.4248\times10^{-1530}     $&$9.3926\times10^{-1530}      $& $0.3410\%$\\
\hline
\end{tabular}
\end{table}

\begin{table}[H]
\centering
\caption{List of numerical results compared with theoretical values}
\begin{tabular}{|cccllr|cccllr|}
\hline
$\alpha$        &$\beta$              &Size $N+1$      & Numerical $\lambda_{N}$ & Theoretical $\lambda_{N}$    & error     \\  \hline
$      0       $&$         0         $&  100         &$5.7797\times10^{-151}     $&$5.7827\times10^{-151}      $& $0.0519\%$\\
                &                     &  200         &$6.3585\times10^{-304}     $&$6.3601\times10^{-304}      $& $0.0261\%$\\
                &                     &  300         &$6.0468\times10^{-457}     $&$6.0478\times10^{-457}      $& $0.0174\%$\\
                &                     &  400         &$5.4190\times10^{-610}     $&$5.4197\times10^{-610}      $& $0.0131\%$\\
                &                     &  500         &$4.7013\times10^{-763}     $&$4.7018\times10^{-763}      $& $0.0105\%$\\
                &                     &  600         &$3.9959\times10^{-916}     $&$3.9963\times10^{-916}      $& $0.0087\%$\\
                &                     &  700         &$3.3487\times10^{-1069}     $&$3.3489\times10^{-1069}      $& $0.0075\%$\\
                &                     &  800         &$2.7774\times10^{-1222}     $&$2.7776\times10^{-1222}      $& $0.0065\%$\\
                &                     &  900         &$2.2855\times10^{-1375}     $&$2.2856\times10^{-1375}      $& $0.0058\%$\\
                &                     & 1000         &$1.8690\times10^{-1528}     $&$1.8691\times10^{-1528}      $& $0.0052\%$\\
$      1       $&$         2         $&  100         &$1.1930\times10^{-152}     $&$1.1683\times10^{-152}      $& $2.0716\%$\\
                &                     &  200         &$1.2985\times10^{-305}     $&$1.2849\times10^{-305}      $& $1.0458\%$\\
                &                     &  300         &$1.2304\times10^{-458}     $&$1.2218\times10^{-458}      $& $0.6994\%$\\
                &                     &  400         &$1.1101\times10^{-611}     $&$1.0949\times10^{-611}      $& $0.5254\%$\\
                &                     &  500         &$9.5390\times10^{-765}     $&$9.4989\times10^{-765}      $& $0.4207\%$\\
                &                     &  600         &$8.1019\times10^{-918}     $&$8.0735\times10^{-918}      $& $0.3508\%$\\
                &                     &  700         &$6.7861\times10^{-1071}     $&$6.7657\times10^{-1071}      $& $0.3009\%$\\
                &                     &  800         &$5.6262\times10^{-1224}     $&$5.6114\times10^{-1224}      $& $0.2633\%$\\
                &                     &  900         &$4.6283\times10^{-1377}     $&$4.6175\times10^{-1377}      $& $0.2341\%$\\
                &                     & 1000         &$3.7840\times10^{-1530}     $&$3.7760\times10^{-1530}      $& $0.2108\%$\\
$     2        $&$         1         $&  100         &$6.0492\times10^{-153}     $&$5.8413\times10^{-153}      $& $3.4369\%$\\
                &                     &  200         &$6.5383\times10^{-306}     $&$6.4245\times10^{-306}      $& $1.7405\%$\\
                &                     &  300         &$6.1811\times10^{-459}     $&$6.1091\times10^{-459}      $& $1.1653\%$\\
                &                     &  400         &$5.5230\times10^{-612}     $&$5.4746\times10^{-612}      $& $0.8758\%$\\
                &                     &  500         &$4.7830\times10^{-765}     $&$4.7494\times10^{-765}      $& $0.7016\%$\\
                &                     &  600         &$4.0605\times10^{-918}     $&$4.0367\times10^{-918}      $& $0.5851\%$\\
                &                     &  700         &$3.3999\times10^{-1071}     $&$3.3829\times10^{-1071}      $& $0.5019\%$\\
                &                     &  800         &$2.8181\times10^{-1224}     $&$2.8057\times10^{-1224}      $& $0.4393\%$\\
                &                     &  900         &$2.3178\times10^{-1377}     $&$2.3087\times10^{-1377}      $& $0.3907\%$\\
                &                     & 1000         &$1.8947\times10^{-1530}     $&$1.8880\times10^{-1530}      $& $0.3517\%$\\
$      10      $&$        10         $&  100         &$7.7362\times10^{-163}     $&$2.8934\times10^{-163}      $& $62.5995\%$\\
                &                     &  200         &$5.3174\times10^{-316}     $&$3.1823\times10^{-316}      $& $40.1536\%$\\
                &                     &  300         &$4.2830\times10^{-469}     $&$3.0260\times10^{-469}      $& $29.3481\%$\\
                &                     &  400         &$3.5259\times10^{-622}     $&$2.7118\times10^{-622}      $& $23.0900\%$\\
                &                     &  500         &$2.9052\times10^{-775}     $&$2.3526\times10^{-775}      $& $19.0220\%$\\
                &                     &  600         &$2.3852\times10^{-928}     $&$1.9995\times10^{-928}      $& $16.1690\%$\\
                &                     &  700         &$1.9497\times10^{-1081}     $&$1.6756\times10^{-1081}      $& $14.0587\%$\\
                &                     &  800         &$1.5871\times10^{-1234}     $&$1.3898\times10^{-1234}      $& $12.4348\%$\\
                &                     &  900         &$1.2871\times10^{-1387}     $&$1.1436\times10^{-1387}      $& $11.1468\%$\\
                &                     & 1000         &$1.0403\times10^{-1540}     $&$9.3519\times10^{-1541}      $& $10.1004\%$\\
\hline
\end{tabular}
\end{table}

\newpage

From Tables 1-3 and Figures $1,2$, We note that the Theoretical $\lambda_{N}$ is slightly less than the Numerical $\lambda_{N}$ for the the cases  $\alpha\beta>0$ and $\alpha\beta<0$. And the other way around for $\alpha=\beta=0$. What's more, for the case $\alpha<0,\beta=0$, we can find an interesting point $\alpha=-\frac{1}{2},\beta=0$ where the errors for virous $N$ approximate to zero, see Figure $2$.

\begin{rem}\label{rem1}
In Figure 2,
\begin{equation}\label{2.29}
\%~ error=\frac{Theoretical~\lambda_{N}-~Numerical~\lambda_{N}}{Numerical~\lambda_{N}}\times100 ,
\end{equation}
and the values $\%~ error$ in Figure 1,3,4 and Table 1-3 are given by the absolute values of (\ref{2.29}).
\end{rem}

\begin{rem}\label{rem1}
Figure 3 and 4 (given by 576 points) show that the contour lines of the error have elliptical shapes, because the value of the integral (b) is approximately double the value of (c), which can be seen from (2.14) and (2.15).
\end{rem}
\section{Appendix}
The integrals identities listed below, can be found in [\cite{new1}] and [\cite{new2}].

For $t\geq 1$,
\begin{equation}
\int_{0}^{1}\frac{dx}{(x+t)\sqrt{x(1-x)}}=
\frac{\pi}{\sqrt{t(t+1)}}.
\end{equation}
\begin{equation}
\int_{0}^{1}\frac{\ln(1-x)}{(x+t)\sqrt{x(1-x)}}dx=
\frac{\pi\ln \left[\frac{t+1}{\left(\sqrt{t}+\sqrt{t+1}\right)^{2}}\right]}{\sqrt{t(t+1)}}.
\end{equation}

For $t<-1$ (here, we take $\sqrt{c}=\rm{i}\sqrt{-c}, c<0$, with $\rm{i}^{2}=-1$),
\begin{equation}
\int_{0}^{1}\frac{dx}{(x+t)\sqrt{x(1-x)}}=
-\frac{\pi}{\sqrt{t(t+1)}}.
\end{equation}
\begin{equation}
\int_{0}^{1}\frac{\ln(1-x)}{(x+t)\sqrt{x(1-x)}}dx=
-\frac{\pi\ln \left[\frac{t+1}{\left(\sqrt{t}+\sqrt{t+1}\right)^{2}}\right]}{\sqrt{t(t+1)}}.
\end{equation}

\section{Acknowledgements}
The financial support of the Macau Science and Technology Development Fund under grant
number FDCT 130/2014/A3 and FDCT 023/2017/A1 are gratefully acknowledged. We would also like
to thank the National Science Foundation (NSF): CCF-1525754 and the University of Macau for generous support: MYRG 2014-00011 FST, MYRG
2014-00004 FST.


\begin{thebibliography}
\small
\bibitem{C1}
Y. L. Geronimus, Polynomials orthogonal on a circle and interval, Pergamon, New York, 1960.
\bibitem{C2}
G. Szeg\"{o}, Orthogonal polynomials, rev. ed., Amer. Math. Soc. Colloq. Publ., Vol. 23, Amer. Math. Soc., Providence, R. I., 1959.
\bibitem{C3}
G. Szeg\"{o}, On some Hermitian forms associated with two given curves of the complex plane, Trans. Amer. Math. Soc. 40 (1936), 450-461. In: Collected papers (volume 2), 666-678. Birkha\"{u}ser, Boston, Basel, Stuttgart, 1982.
\bibitem{C27}
G. P\'{o}lya, G. Szeg\"{o}, Problems and and Theorems in Analysis I, Springer-Verlag, Berlin Heidelberg, New York, 1978.
\bibitem{C4}
G. Szeg\"{o}, Hankel forms (English translation of A Hankel-f\'{e}le form\'{a}kr\'{o}l) Collected papers, 1 (1982)(Basle:Birkh\"{a}user) P111.
\bibitem{C5}
J. Todd, Contributions to the solution of systerms of linear equations and the determination of eigenvalues, Nat. Bur. Standards Appl. Math. Ser. 39 (1959), 109-116.
\bibitem{C6}
H. Widom, H. S. Wilf, Small eigenvalues of large Hankel matrices, Proc. Amer. Math. Soc. 17 (1966), 338-344.
\bibitem{new3}
H. Widom, H. S. Wilf, Errata: Small Eigenvalues of Large Hankel Matrices, Proc. Amer. Math. Soc. 19 (1968), 1508.
\bibitem{C7}
H. S. Wilf, Finite sections of some classical inequalities, springer, Berlin, Heidelberg, New York, 1970.
\bibitem{C8}
C. A. Tracy, H. Widom, The distributions of random matrix theory and their applications// New trends in Mathematical physics. Netherlands: Springer, 2009, 753-765.
\bibitem{C9}
Y. Chen, N. D. Lawrence, Small eigenvalues of large Hankel matrices, J. Phys. A: Math. Gen., 32 (1999), 7305-7315.
\bibitem{C10}
Y. Chen, D. S. Lubinsky, Smallest eigenvalues of Hankel matrices for exponential weights, J. Math. Anal. Appl. 293 (2004), 476-495.
\bibitem{new1}
Y. Chen, M. R. McKay, Coulumb fluid, Painlev\'{e} transcendents, and the information theory
of MIMO systems, IEEE Trans. Inform. Theory 58 (2012), 4594-4634.
\bibitem{C11}
N. Emmart, Y. Chen, C. C. Weems, Computing the smallest eigenvalue of large ill-conditioned Hankel matrices, Commun. Comput. Phys. 18 (2015), 104-124.
\bibitem{C12}
S. Lyu, Y. Chen, The largest eigenvalue distribution of the Laguerre unitary ensemble. Acta. Math. Sci. 37B (2017), 1-24.
\bibitem{C13}
C. Berg, Y. Chen, M. E. H. Ismail, Small eigenvalues of large Hankel matrices: The indeterminate case, Math. Scand., 91 (2002), 67-81.
\bibitem{C14}
C. Berg, R. Szwarc, The smallest eigenvalue of Hankel matrices, Constructive Approximation, 34 (2011), 107-133.
\bibitem{C15}
C. Berg, R. Szwarc, Symmetric moment problems and a conjecture of Valent. Sbornik: Mathematics, 208 (2017), 335-359.
\bibitem{C16}
N. I. Akhiezer, The classical moment problem and some related questions in analysis, English translation, Oliver and Boyd, Edinburgh, 1965.
\bibitem{C17}
M. L. Mehta, Random Matrices, 3rd edition. Elsevier, Singapore, 2006.
\bibitem{C18}
N. J. Higham, Functions of matrices: Theory and computation. Philadelphia: SIAM, 2008.
\bibitem{C19}
S. T. M. Ackermans, An asymptotic method in the theory of series (Thesis). 83 PP, Eindhofen, 1964.
\bibitem{new2}
I. S. Gradshteyn, I. M. Ryzhik, Table of integrals, series, and products, seventh ed. (Elsevier/
Academic Press, Amsterdam, 2007) pp. xlviii+1171, translated from the Russian, Translation
edited and with a preface by Alan Jeffrey and Daniel Zwillinger, With one CD-ROM
(Windows, Macintosh and UNIX).
\end{thebibliography}
\end{document}